\documentclass[11pt]{article}
\usepackage{times}
\usepackage{url}
\usepackage[a4paper]{geometry}
\usepackage{comment}

\usepackage{subfig}
\usepackage{amsmath}
\usepackage{amsthm}
\usepackage{amssymb}
\usepackage{wrapfig}
\usepackage{graphicx}
\graphicspath{{./Figs/}}
\usepackage{algorithm}

\usepackage{paralist}
\usepackage[active]{srcltx}
\usepackage{setspace}%\singlespacing
\usepackage{fullpage}

\newenvironment{proof sketch}{\noindent {\bf Proof sketch:} }{\hfill \qed}

\newtheorem{theorem}{Theorem}
\newtheorem{proposition}{Proposition}
\newtheorem{lemma}{Lemma}
\newtheorem{coro}[theorem]{Corollary}
\newtheorem{rem}{Remark}

\newtheorem{assumption}{Assumption}

\newcommand{\eps}{\varepsilon}

\newcommand{\sinr}{\text{\sc{sinr}}}
\newcommand{\snr}{\text{\sc{snr}}}

\newcommand{\opt}{\text{\textsc{opt}}}

\newcommand{\sector}{\text{\emph{sector}}}

\newcommand{\disperse}{\text{\emph{disperse}}}

\newcommand{\LL}{\mathcal{L}}
\newcommand{\FF}{\mathcal{F}}
\newcommand{\eqdf}{\triangleq}
\newcommand{\MAXMIN}{\text{\sc{Max-Min Throughput}}}
\newcommand{\MAX}{\text{\sc{Max Throughput}}}
\newcommand{\RR}{{\mathbb R}}

\newcommand{\XTH}{F}
\newcommand{\XNTH}{R}
\newcommand{\LPmax}{\textsc{MaxTh}_{LP}}
\newcommand{\LPmin}{\textsc{MaxMinTh}_{LP}}
\newcommand{\pmin}{P_{\min}}
\newcommand{\pmax}{P_{\max}}

\newcommand{\Cbad}{C^{\text{bad}}_u}
\newcommand{\Tzo}{\{0,\ldots, T-1\}}
\newcommand{\gcolor}{\text{\emph{greedy-coloring}}}
\newcommand{\At}{A^{\Delta}}
\newcommand{\sigdiv}{\sigma}

\newcommand{\barsignal}{$\text{signal}$}

\newcommand{\capslot}{\text{\sc cap-1slot}}
\newcommand{\lathop}{\text{\sc lat-1hop}}
\newcommand{\latpaths}{\text{\sc lat-paths}}
\newcommand{\latroute}{\text{\sc lat-route}}
\newcommand{\throute}{\text{\sc throughput-route}}

\begin{document}
%%% ----------------------------------------------------------------------
%\begin{frontmatter}

\title{Multi-Hop Routing and Scheduling in Wireless Networks in the SINR Model}

\author{%
Guy Even
\thanks{School of Electrical Engineering,
Tel-Aviv Univ., Tel-Aviv 69978, Israel.
\protect\url{{eveng,yakovmat,medinamo}@post.tau.ac.il}}
\and
Yakov Matsri~$^*$
\and
Moti Medina~$^*$
\date{}
%\institute{%
%School of Electrical Engineering,Tel-Aviv Univ., Tel-Aviv 69978, Israel.
%\email{\{eveng,yakovmat,medinamo\}@post.tau.ac.il}
%}
}
\maketitle \date{}

%%% ----------------------------------------------------------------------

\begin{abstract}
  We present an algorithm for multi-hop routing and scheduling of
  requests in wireless networks in the \sinr\ model. The goal of our
  algorithm is to maximize the throughput or maximize the minimum
  ratio between the flow and the demand.

  Our algorithm partitions the links into buckets.  Every bucket
  consists of a set of links that have nearly equivalent reception
  powers.  We denote the number of nonempty buckets by $\sigdiv$.  Our
  algorithm obtains an approximation ratio of $O(\sigdiv \cdot \log
  n)$, where $n$ denotes the number of nodes.  For the case of linear
  powers $\sigdiv =1$, hence the approximation ratio of the algorithm
  is $O(\log n)$.  This is the first practical approximation algorithm
  for linear powers with an approximation ratio that depends only on
  $n$ (and not on the max-to-min distance ratio).

  If the transmission power of each link is part of the input (and
  arbitrary), then $\sigdiv = O(\log\Gamma + \log \Delta)$, where
  $\Gamma$ denotes the ratio of the max-to-min power, and $\Delta $
  denotes the ratio of the max-to-min distance.  Hence, the
  approximation ratio is $O(\log n \cdot (\log\Gamma + \log \Delta))$.

  Finally, we consider the case that the algorithm needs to assign
  powers to each link in a range $[\pmin,\pmax]$. An extension of the
  algorithm to this case achieves an approximation ratio of $O[(\log n
  + \log \log \Gamma) \cdot (\log\Gamma + \log \Delta)]$.

\end{abstract}
\begin{comment}

  \paragraph{ Keywords:}
\end{comment}

%\end{frontmatter}

%\paragraph{Organization.}

%%% ----------------------------------------------------------------------

\section{Introduction}
In this paper we deal with the problem of maximizing throughput in a
wireless network.  Throughput is a major performance criterion in many
applications, including: file transfer and video streaming.  It has
been acknowledged that efficient utilization of network resources
requires so called cross layered algorithms~\cite{lin2006tutorial}.
This means that the algorithm deals with tasks that customarily belong
to different layers of the network. These tasks include: routing,
scheduling, management of queues in the nodes, congestion control, and flow control.

The problem we consider is formulated as follows.  We are given a set
$V$ of $n$ nodes in the plane.  A link $e$ is a pair $(s_e,r_e)$ of
nodes with a power assignment $P_e$. The node $s_e$ is the transmitter
and the node $r_e$ is the receiver. In the \sinr\ model, $r_e$
receives a signal from $s_e$ with power $S_e=P_e / d_e^\alpha$, where
$d_e$ is the distance between $s_e$, and $r_e$ and $\alpha$ is the
path loss exponent.  The network is given a set of requests
$\{R_i\}_{i=1}^k$.  Each request is a $3$-tuple $R_i=({\hat s}_i,{\hat
  t}_i,b_i)$, where ${\hat s}_i \in V$ is the source, ${\hat t}_i \in
V$ is the destination, and $b_i$ is the requested packet rate.  The
output is a multi-commodity flow $f=(f_1,\ldots,f_k)$ and an
\sinr-schedule $S=\{L_t\}_{t=0}^{T-1}$ that supports $f$.  Each $L_t$
is a subset of links that can transmit simultaneously
(\sinr-feasible).  The goal is to maximize the total flow
$|f|=\sum_{i=1}^k |f_i|$. We also consider a version that maximizes
$\min_{i=1\ldots k} |f_i|/b_i$.  Let $\Delta\eqdf {d_{\max}} / d_{\min}$ is
the ratio between the maximum and minimum length of a link, and
$\Gamma\eqdf \pmax/\pmin$ the ratio between the maximum and minimum
transmission power.  For the case in which $\max _{e \neq e'}\frac
{S_e}{S_{e'}}=O(1)$, the approximation ratio achieved by the algorithm
is $O(\log n)$.  For arbitrary powers and link lengths, the approximation ratio
achieved by the algorithm is $O(\log n \cdot (\log\Gamma + \log
\Delta))$.

\paragraph{Previous Work.}
Gupta and Kumar~\cite{gupta2000capacity} studied the capacity of
wireless networks in the \sinr-model and the graph model for random
instances in a square.  The \sinr-model for wireless networks was
popularized in the algorithmic community by Moscibroda and Wattenhofer
~\cite{moscibroda2006complexity}. NP-Completeness for scheduling a
set of links was proven by
Goussevskaia~\cite{goussevskaia2007complexity}.

Algorithms for routing and scheduling in the \sinr-model can be
categorized by four main criteria: maximum capacity with one round vs.
scheduling, multi-hop vs. single-hop, throughput maximization vs.
latency minimization, and the choice of transmitter powers.  In the
single-hop setting, routing is not an issue, and the focus is on
scheduling. If the objective is latency minimization, then each
request is treated as a task, and the goal is to minimize the
makespan.

The following problems are considered.
\begin{inparaenum}[(1)]
\item \label{item:prob1}
\capslot : find a subset of maximum cardinality that is SINR-feasible.
\item \label{item:prob2}
\lathop : find a shortest SINR-schedule for a set of links.
\item \label{item:prob3}
\latpaths : find a shortest SINR-schedule for a set of paths.
\item \label{item:prob4}
\latroute : find a routing and a shortest SINR-schedule for a set of multi-hop requests.
\item \label{item:prob5} \throute : find a routing and maximum
  throughput SINR-schedule for a set of multi-hop requests.
\end{inparaenum}
We briefly review some of the algorithmic results in this area
published in the last three years.

Chafekar et al.~\cite{chafekar2007cross} present an approximation
algorithm for \latroute. The
approximation ratio is $O(\log n \cdot \log \Delta \cdot \log^2
\Gamma)$. Fangh{\"a}nel
et al.~\cite{fanghanel2010improved} improved this result to $O(\log
\Delta \cdot \log^2 n)$.  Goussevskaia et
al.~\cite{goussevskaia2009capacity} pointed out that $\log \Delta$ can
be $\Omega(n)$, and presented the first approximation algorithm whose
approximation ratio is always nontrivial. In fact, the approximation
ratio obtained by Goussevskaia et al.~\cite{goussevskaia2009capacity}
is $O(\log n)$ for the case \lathop\ with
uniform power transmissions.
%Shortly after that, Halldorsson and Wattenhofer~\cite{HW} improved the approximation ratio to a constant.

Halldorsson~\cite{halldorsson2009wireless} presented algorithms for
\lathop\ with mean power assignments. He presented an $O(\log n \log \log
\Delta)$-approximation and an $O(\log \Delta)$-online algorithm that
uses mean power assignments with
respect to OPT that can choose arbitrary power assignments (see
also~\cite{tonoyan2010algorithms}).

Halldorsson and Mitra~\cite{halldorssonwireless} presented a constant
approximation algorithm for \capslot\ problem with uniform,
linear and mean power assignments. In addition, by using the mean power assignment,
the algorithm obtains a $O(\log n +\log \log \Delta)$-approximation
with respect to arbitrary power assignments.

Kesselheim and V\"ocking ~\cite{kesselheim2010distributed} give a
distributed randomized algorithm for \lathop\ that obtains an $O(\log^2 n)$-approximation
using uniform and linear powers. Halldorson and Mitra
~\cite{halldorsson2011nearly} improve the analysis to achieve an $O(\log
n)$-approximation.

Kesselheim~\cite{K10} presents approximation results in the
\sinr-model: an $O(1)$-approximation for \capslot, an $O(\log
n)$-approximation for \lathop, an $O(\log^2 n)$-approximation
for \latpaths\ and \latroute.  In ~\cite{K10} there is no limitation
on power assignment imposed neither on the solution nor on the optimal
solution.  In practice, power assignments are limited, especially for
mobile users with limited power supply.

The most relevant work to our result is by Chafekar et
al.~\cite{ChafekarCapacity} who presented approximation algorithms for
\throute. They present the following results, an $O(\log
\Delta)$-approximation for uniform power assignment and linear power
assignment, and an $O(\log \Gamma \cdot \log \Delta)$ for arbitrary
power assignments.

For linear powers, Wan et al.~\cite{wanwireless} obtain a $O(\log
n)$-approximation for \throute.  The algorithm is based on a reduction
to the single-slot problem using the ellipsoid method.
In~\cite{wan2009multiflows}, Wan writes that ``this algorithm is of
theoretical interest only, but practically quite infeasible.''
For the case that the algorithm assigns powers from a limited range,
Wan et al.~\cite{wanwireless} achieve an $O(\log n \cdot \log
\Gamma)$-approximation ratio.

\paragraph{Our result.}
 We present an algorithm for
\throute.  Our algorithm partitions the links into buckets.  Every
bucket consists of a set of links that have nearly equivalent
reception powers.  We denote the number of nonempty buckets (also
called the signal diversity of the links) by $\sigdiv$.  Our algorithm
obtains an approximation ratio of $O(\sigdiv \cdot \log n)$, where $n$
denotes the number of nodes.

For the case of linear power assignment the signal diversity is
$\sigdiv =1$, hence the approximation ratio of the algorithm is
$O(\log n)$.  This is the first practical approximation
algorithm for linear powers that obtains an approximation ratio that
depends only on $n$ (and not on ratio of the max-to-min distance).
This improves the $O(\log \Delta)$-approximation of Chafekar et
al.~\cite{ChafekarCapacity} for linear power assignment.  As pointed
out in~\cite{goussevskaia2009capacity}, $\log \Delta$ can be
$\Omega(n)$.  The linear power assignment model makes a lot of sense
since it implies that, in absence of interferences, transmission
powers are adjusted so that the reception powers are uniform.

In the case of arbitrary given powers, the signal diversity is
$\sigdiv = O(\log\Gamma + \log \Delta)$.  Hence, the approximation
ratio is $O(\log n \cdot (\log\Gamma + \log \Delta))$.  For arbitrary
power assignments Chafekar et al.~\cite{ChafekarCapacity} presented
approximation algorithm that achieves approximation ratio of $O(\log
\Gamma \cdot \log \Delta)$.  In this case, the approximation ratio of
our algorithm is not comparable with the algorithm presented by
Chafekar et al.~\cite{ChafekarCapacity} (i.e., in some cases it is
smaller, in other cases it is larger).
%$O(\log n \cdot (\log\Gamma + \log \Delta))$ is not comparable with
%$O(\log\Gamma \cdot \log \Delta))$.
%We show that for some instances this approximation ratio improves over~\cite{ChafekarCapacity}
%by a factor of $O \left(\frac {\sqrt n}{\log n}\right)$.

For the case of limited powers where the algorithm needs to assign
powers between $\pmin$ and $\pmax$, we give a $O[(\log n + \log \log
\Gamma) \cdot (\log\Gamma + \log \Delta)]$-approximation algorithm.

Our results apply both for maximizing the total throughput and for
maximizing the minimum fraction of supplied demand. Other fairness
criteria apply as well (see also \cite{ChafekarPhD}).

\paragraph{Techniques.}
Similarly to~\cite{ChafekarCapacity} our algorithm is based on linear
programming relaxation and greedy coloring.
The linear programming relaxation determines the routing and the flow along each route.
Greedy coloring induces a schedule in which, in every slot, every link is \sinr-feasible with respect to longer links in the same slot.

We propose a new method of classifying the links.
In~\cite{ChafekarCapacity,halldorsson2009wireless} the links are classified by lengths and by
transmitted powers.
On the other hand, we classify the links by their \emph{received power}.

We present a new linear programming formulation for throughput
maximization in the \sinr-model.  This
formulation uses novel symmetric interference constraints, for every link $e$, that
bound the interference incurred by other links in the same bucket as well as the
interference that $e$ incurs to other links.
We show that this formulation is a relaxation due to our link classification method.

We then apply a greedy coloring procedure for rounding the LP
solution. This method
follows~\cite{alicherry2005joint,ChafekarCapacity,wan2009multiflows}
and others (the greedy coloring is described in Section~\ref{sec:coloring}).

The schedule induced by the greedy coloring is not \sinr-feasible.
Hence, we propose a refinement technique that produces an \sinr-feasible schedule.
We refine each color class using a bin packing procedure that is based
on the symmetry of the interference coefficients in the LP. We believe
this method is of independent interest since it mitigates the problem
of bounding the interference created by
shorter links.

\paragraph{Organization.}
In Sec.~\ref{sec:prelim} we present the definitions and notation.  The
throughput maximization problem is defined in Sec.~\ref{sec:problem}.
%The linear programming formulation is presented in Sec.~\ref{sec:LP}.
In Sec.~\ref{sec:necessary}, we present necessary conditions for
\sinr-feasibility for links that are in the same bucket.  The results
in Sec.~\ref{sec:necessary} are used for proving that the linear
programming formulation presented in Sec.~\ref{sec:LP} is indeed a
relaxation of the throughput maximization problem.  The algorithm for
linear powers is presented in Sec.~\ref{sec:alg} and analyzed in
Sec.~\ref{sec:analysis}. In Sec.~\ref{sec:arbitraty} we extend the algorithm so that it handles arbitrary powers.
In Sec.~\ref{sec:limited} we extend the algorithm so that it assigns limited powers.

\section{Preliminaries}
\label{sec:prelim}
We briefly review definitions used in the literature for algorithms in
the \sinr\ model (see~\cite{HW,ChafekarCapacity}).

We consider a wireless network that consists of a set $V$ of $n$ nodes in the plane.
Each node is equipped with a transmitter and a receiver.
We denote the distance between nodes $u$ and $v$ by $d_{uv}$.

A \emph{link} is a $3$-tuple $e=(s_e,r_e,P_e)$, where $s_e\in V$ is the
transmitter, $r_e\in V$ is the receiver, and $P_e$ is the transmission power.
In the general setting we allow parallel links with different powers.
The set of links is denoted by $\LL$ and $m \eqdf |\LL|$.
We abbreviate and denote the distance $d_{s_er_e}$ by $d_e$.
Similarly, we denote the distance $d_{s_er_e'}$ by $d_{ee'}$.
Note that according to this notation, $d_{ee'}\neq d_{e'e}$.

We use the following radio propagation model.  A transmission from
point $s$ with power $P$ is received at point $r$ with power
$P/d_{sr}^{\alpha}$.  The exponent $\alpha$ is called the \emph{path
  loss exponent} and is a constant.
  In most practical
situations, $2 \leq \alpha \leq 6$; our algorithm works for any constant $\alpha \geq 0$.
For links $e,e'$, we use the
following notation: $S_e\eqdf P_e/d_e^\alpha$ and
$S_{e'e} \eqdf P_{e'}/{d_{e'e}^{\alpha}}$.

A subset of links $L\subseteq \LL$ is \sinr-feasible if
$S_e/(N+\sum_{e'\in L \setminus\{e\}} S_{e'e}) \geq \beta$, for every $e\in L$.
This ratio is called the \emph{signal-to-noise-interference ratio}
(\sinr), where the constant $N$ is positive and models the noise in
the system.  The threshold $\beta$ is a positive constant. The ratio
$S_e/N$ is called the \emph{signal-to-noise ratio} (\snr).

A link $e$ can tolerate an accumulated interference $\sum_{e'} S_{e'e}$ that is at most
$(S_e-\beta N)/\beta$.  This amount can be considered to be the
``interference budget'' of $e$.  Let $\gamma_e \eqdf (\beta S_e)/(S_e-\beta N)$.
We define three measures of how much of the interference budget is ``consumed'' by a link $e'$.
\begin{align*}
  \hat a_{e'}(e)&\eqdf \frac{S_{e'e}}{S_e},
&
  a_{e'}(e)&\eqdf \gamma_e \cdot   \hat a_{e'}(e), \text{ and }
&
  \bar a_{e'}(e)&\eqdf \min\{1,a_{e'}(e)\}.
\end{align*}
The value of $a_{e'}(e)$ is called the \emph{affectance}~\cite{HW} of the link
$e'$ on the link $e$. The affectance is additive, so
for a set $L\subseteq\LL$, let $a_L(e) \eqdf \sum_{\{e'\in L : e'\neq e\}} a_{e'}(e)$.
\begin{proposition}[\cite{HW}]
  A set $L\subseteq \LL$ is \sinr-feasible iff $a_L(e) \leq 1$, for every $e\in L$.
\end{proposition}

Following~\cite{HW}, we define a set $L\subseteq \LL$ to be a
$p$-signal, if $a_L(e) \leq 1/p$, for every $e\in L$. Note that $L$ is
\sinr-feasible if $L$ is a $1$-signal.
We also define a set $L\subseteq \LL$ to be a
$\bar p$-\barsignal, if ${\bar a}_L(e) \leq 1/p$, for every $e\in L$. Note that $L$ is
\sinr-feasible if $L$ is a $\overline{(1+\eps)}$-\barsignal\ for some $\eps > 0$.

By Shannon's theorem on the capacity of a link in an additive white
Gaussian noise channel~\cite{gallager1968information}, it follows that the capacity
is a function of the \sinr. Since we use the same threshold $\beta$
for all the links, it follows that
the Shannon capacity of a link is either zero (if the \sinr\ is less than $\beta$)
or a value determined by $\beta$ (if the \sinr\ is at least $\beta$).
We set the length of a time slot and a packet length so that, if
interferences are not too large, each link can deliver one packet in one
time slot.  By setting a unit of flow to equal a packet-per-time-slot, all
links have unit capacities.
We do not assume that $\beta \geq 1$; in fact,
in communications systems $\beta$ may be smaller than one.

\paragraph{Multi-commodity flows.}
Recall that a function $g: {\cal L} \rightarrow {\mathbb{R}}^{\geq 0}$
is a flow from $s$ to $t$, where $s,t \in V$, if it satisfies capacity
constraints (i.e., $g(e) \leq 1$, for every $e \in {\cal L}$) and flow conservation
constraints in every vertex $v \in V \setminus \{s,t\}$ (i.e., $\sum_{e \in \text{in}(v)}g(e)=\sum_{e\in \text{out}(v)}g(e)$).

We use multi-commodity flows to model multi-hop traffic in a network.
The network consists of the nodes $V$ and the arcs $\LL$, where
each arc has a unit capacity.  There are $k$ commodities
$R_i=({\hat s}_i,{\hat t}_i,b_i)$, where ${\hat s}_i$ and ${\hat t}_i$ are the \emph{source} and
\emph{sink}, and $b_i$ is the \emph{demand} of the $i$th commodity.
Consider a vector $f=(f_1,\ldots,f_k)$, where each
$f_i$ is a flow from ${\hat s}_i$ to ${\hat t}_i$.
We use the following notation:
\begin{inparaenum}[(i)]
\item $f_i(e)$ denotes the flow of the $i$th flow along $e$,
\item $|f_i|$ equals the amount of flow shipped from ${\hat s}_i$ to ${\hat t}_i$,
\item $f(e)\eqdf\sum_{i=1}^k f_i(e)$,
\item $|f| \eqdf \sum_{i=1}^k |f_i|$.
\end{inparaenum}
A vector $f=(f_1,\ldots,f_k)$ is a multi-commodity flow if $f(e) \leq 1$,
for every $e \in \LL$.

We denote by $\FF$ the polytope of all multi-commodity flows $f=(f_1,\ldots,f_k)$
such that $|f_i|\leq b_i$, for every $i$.
For a $\rho>0$, we denote by $\FF_\rho\subseteq \FF$ the polytope of
all multi-commodity flows such that $|f_i|/b_i\geq \rho$.

\paragraph{Schedules and multi-commodity flows.}
We use periodic schedules to support a multi-commodity flow using
packet routing as follows.  We refer to a sequence
$\{L_t\}_{t=0}^{T-1}$, where $L_t\subseteq\LL$ for each $i$, as a
\emph{schedule}. A schedule is used periodically to determine which
links are active in each time slot.  Namely, time is partitioned into
disjoint equal time slots.  In time slot $t'$, the links in $L_t$, for
$t=t' \pmod T$ are \emph{active}, namely, they transmit.  Each active
link transmits one packet of fixed length in a time slot (recall that
all links have the same unit capacity).  The number of time slots $T$
is called the \emph{period} of the schedule. We sometimes represent a
schedule $S=\{L_t\}_{t=0}^{T-1}$ by a multi-coloring $\pi:\LL
\rightarrow 2^{\{0,\ldots,T-1\}}$.  The set $L_t$ simply equals the
preimage of $t$, namely, $L_t=\pi^{-1}(t)$, where
$\pi^{-1}(t)\eqdf\{e: t\in \pi(e)\}$.

An \emph{\sinr-schedule} is a sequence
$\{L_t\}_{t=0}^{T-1}$ such that $L_t$ is \sinr-feasible for every $t$.
Consider a multi-commodity flow $f=(f_1,\ldots,f_k)$ and a
schedule $S=\{L_t\}_{t=0}^{T-1}$. We say that the schedule $S$
\emph{supports} $f$ if
\begin{align*}
\forall e\in \LL:~~~  T \cdot f(e) \leq \left|\{t\in \{0,\ldots,T-1\} : e\in L_t\}\right|.
\end{align*}
%We denote by $\support(S)$ the set of all multiflows that are supported by $S$.

The motivation for this definition is as follows. Consider a store-and-forward packet
routing network that schedules links according to the schedule $S$.
This network can
deliver packets along each link $e$ at an average rate of $f(e)$
packets-per-time-slot.

\paragraph{Buckets and signal diversity.}
We partition the links into buckets by their received power $S_e$ . Let $S_{\min} \eqdf \min_{e\in \LL} S_e$.
The $i$th bucket $B_i$ is defined by

\begin{align*}
  B_i &\eqdf \left\{e\in \LL \mid 2^i \cdot S_{\min} \leq S_e <2^{i+1}\cdot S_{\min}   \right\}.
\end{align*}
For a link $e\in \LL$, define $i(e)\eqdf\lfloor\log_2 (S_e/S_{\min})
\rfloor$.  Then, $e\in B_{i(e)}$.  The \emph{signal diversity} $\sigdiv$ of $\LL$
is the number of nonempty buckets.

\begin{lemma}\label{lem:dive}
    $$\sigdiv \leq \lceil \alpha \cdot \log_2\Delta + \log_2 \Gamma \rceil\:.$$
\end{lemma}

\begin{proof}
    Recall that $S_e\eqdf P_e/d_e^\alpha$.
    The signal diversity of $\LL$ is at most
    $\log_2 (S_{\max}/S_{\min})$, where $S_{\max}= \max \{S_e : e \in \LL\}$ and $S_{\min} = \min \{S_e : e \in \LL\}$.
    Hence,
    \begin{eqnarray*}
        \log_2 (S_{\max}/S_{\min}) & \leq & \log_2 \left(\frac{P_{\max}}{d^{\alpha}_{\min}} / \frac{P_{\min}}{d^{\alpha}_{\max}}\right) \\
            & = & \log_2 (\Gamma \cdot \Delta^{\alpha})\:,
    \end{eqnarray*}
    where $P_{\min} = \min \{P_e : e \in \LL\}, P_{\max} = \max \{P_e : e \in \LL\}, d_{\max} = \max \{d_e : e \in \LL\}, d_{\min} = \min \{d_e : e \in \LL\}$, as required.
\end{proof}

\paragraph{Power assignments.}
In the \emph{uniform power assignment}, all links transmit with the same
power, namely, $P_e=P_{e'}$ for every two links $e$ and $e'$.
In the \emph{linear power assignment}, all links receive with the same
power, namely, $S_e=S_{e'}$ for every two links $e$ and $e'$.

\paragraph{Assumption on \snr.} Our analysis requires that, for every
link $e$, $S_e/N\geq (1+\eps)\cdot \beta$, for a constant $\eps>0$.
Note that if $S_e/N=\beta$, then the link cannot tolerate any
interference at all, and $\gamma_e=\infty$.  Our assumption implies
that $\gamma_e \leq (1+\eps)\cdot \beta/\eps$.  This assumption can be
obtained by increasing the transmission power of links whose \snr\
almost equals $\beta$.  Namely, if $S_e/N \approx \beta$, then
$P_e\gets (1+\eps)\cdot P_e$.  A similar assumption is used
in~\cite{ChafekarCapacity}, where it is stated in terms of a
bi-criteria algorithm. Namely, the algorithm uses transmission powers
that are greater by a factor of $(1+\eps)$ compared to the transmission power
of the optimal solution.
\begin{assumption}\label{assume:gamma}
  For every link $e\in \LL$, $S_e/N \geq (1+\eps)\cdot \beta$.
\end{assumption}
\begin{proposition}\label{prop:gamma}
  Under Assumption~\ref{assume:gamma}, $\beta < \gamma_e \leq (1+\eps)\cdot \beta/\eps$.
\end{proposition}
\begin{proof}
    Recall that $\gamma_e \eqdf \frac{\beta S_e}{S_e - \beta N}= \frac{\beta}{1-\beta(N/S_e)}$.
    Assumption~\ref{assume:gamma} implies that  $S_e/N >\beta$.
    Hence, $\gamma_e > \beta$.

    Assumption~\ref{assume:gamma} implies that  $\beta \frac {N}{S_e} \leq \frac 1 {1+\eps}$.
    Hence,
    \begin{eqnarray*}
        \gamma_e    &=     & \frac{\beta}{1-\beta(N/S_e)} \\
                    & \leq & \frac{\beta}{1-\frac 1 {1+\eps}} \\
                    & = & (1+\eps)\cdot \beta / \eps\:,
    \end{eqnarray*}
    as required.
\end{proof}

\section{Problem Definition}\label{sec:problem}
The problem \MAX\ is formulated as follows.  The input consists of:
\begin{inparaenum}[(i)]
    \item A set of nodes $V$ in $\RR ^2$
    \item A set of links $\LL$. The capacity of each link equals one packet per time-slot.
    \item A set of requests $\{R_i\}_{i=1}^k$.  Each request is a
      $3$-tuple $R_i=({\hat s}_i,{\hat t}_i,b_i)$, where ${\hat s}_i \in V$ is the source,
      ${\hat t}_i \in V$ is the destination, and $b_i$ is the requested
      packet rate. We assume that every request can be routed, namely,
      there is a path from ${\hat s}_i$ to ${\hat t}_i$, for every $i\in[1..k]$.
      Since the links have unit capacities, we assume that the
      requested packet rate satisfies $b_i\leq n$.
\end{inparaenum}
The output is a multi-commodity flow $f=(f_1,\ldots,f_k) \in \FF$ and an
\sinr-schedule $S=\{L_t\}_{t=0}^{T-1}$ that supports $f$.
The goal is to maximize the total flow $|f|$.

The \MAXMIN\ problem has the same input and output. The goal, however,
is to maximize $\rho$, such that $f\in \FF_\rho$. Namely, maximize $\min_{i=1\ldots k} |f_i|/b_i$.

\begin{comment}
We note that the \MAX\ problem has two versions: one in which
$|f_i|\leq b_i$, for each $i$, and the second where $|f_i|$ is not
bounded (hence, the demand $b_i$ is ignored). Our algorithm applies to
both versions with the same results. Hence, we focus on the version
without the upper bound (i.e., demands are ignored).
\end{comment}

\section{Necessary Conditions: \sinr-feasibility for links in the same bucket}
\label{sec:necessary}
In this section we formalize necessary conditions so that a set of links in the same bucket is \sinr-feasible. In Section~\ref{sec:LP} we use these conditions to build a LP-relaxation for the problem.

We begin by expressing $\hat a_{e_1}(e_2)$ in terms of the distances $d_{e_1}, d_{e_2}, d_{e_1e_2}$.
Note that $\hat a_{e_1}(e_2)$, with respect to links that are in the same bucket, depends solely on
$d_{e_1}$ and $d_{e_1e_2}$. On the other hand, $\hat a_{e_1}(e_2)$, with respect to the uniform power model, depends solely on $d_{e_2}$ and $d_{e_1e_2}$. The proof of the following proposition is in Appendix~\ref{sec:proofs}.
\begin{proposition}\label{prop:aff}
%  \begin{eqnarray*}
    $$\forall i ~\forall ~e_1,e_2 \in B_i:~~  \frac 12 \cdot \left(\frac{d_{e_1}}{d_{e_1
          e_2}}\right)^\alpha < \hat a_{e_1}(e_2)  <2 \cdot \left(\frac{d_{e_1}}{d_{e_1
          e_2}}\right)^\alpha \:,$$
    $$\forall ~e_1,e_2 \in \LL:~~ \hat a_{e_1}(e_2) = \left(\frac{d_{e_2}}{d_{e_1
          e_2}}\right)^\alpha \text{ in the uniform power model.}$$
%  \end{eqnarray*}
\end{proposition}

%\begin{proof}
%
%    Recall that $\hat a_{e'}(e) \eqdf \frac{S_{e'e}}{S_e}$, $S_e\eqdf P_e/d_e^\alpha$, and
%    $S_{e'e}=P_{e'}/{d_{e'e}^{\alpha}}$.
%    Note that for every two links $e_1,e_2 \in B_i$ satisfy that $S_{e_1} / S_{e_2} \in (1/2,2)$.
%    Hence,
%    \begin{eqnarray*}
%        \hat a_{e_1}(e_2)  & = & \frac{S_{e_1e_2}}{S_{e_2}} = \frac{S_{e_1e_2}}{S_{e_1}} \cdot
%                                    \frac{S_{e_1}}{S_{e_2}}\\
%                        & = & \frac {P_{e_1}/d_{e_1e_2}^{\alpha}}{P_{e_1}/d_{e_1}^{\alpha}} \cdot
%                            \frac{S_{e_1}}{S_{e_2}} \\
%                        & = &\left(\frac {d_{e_1}}{{d_{e_1e_2}}}\right)^{\alpha} \cdot \frac{S_{e_1}}{S_{e_2}}\:,
%    \end{eqnarray*}
%    as required.
%
%    On the other hand, in the uniform power model assignment, all links transmit with the
%    same power, namely $P_e = P_{e'}$ for every two links $e$ and $e'$. Hence,
%    \begin{eqnarray*}
%        \hat a_{e_1}(e_2)  & = & \frac{S_{e_1e_2}}{S_{e_2}} \\
%                        & = & \frac {P_{e_1}/d_{e_1e_2}^{\alpha}}{P_{e_2}/d_{e_2}^{\alpha}} \\
%                        & = &\left(\frac {d_{e_2}}{{d_{e_1e_2}}}\right)^{\alpha}\:,
%    \end{eqnarray*}
%    as required.
%
%\end{proof}

Throughout this section we assume the following.
Let $L \subseteq \LL$ denote an \sinr-feasible set of links such that all the links in $L$
belong to same bucket $B_i$.
Let $e\in B_i$ denote an arbitrary link (not
necessarily in $L$).
\paragraph{Notation.}
Define:
\begin{align*}
  L^{\ell} & \eqdf  \{e' \in L : d_{e'} \leq  d_{e'e}\}, \text{ and}\\
  L^g & \eqdf  \{e' \in L : d_{e'} > d_{e'e}\}.
\end{align*}
\medskip

\subsection{A Geometric Lemma}
The following lemma claims that for every $e \in  B_i$ (not necessarily in $L$), there exits a set of
at most six ``guards'' that ``protect'' $e$ from interferences by transmitters in
$L^{\ell}$.

\begin{lemma}\label{lem:guards}
  There exists a set $G$ of at most six receivers of links in
  $L^{\ell}$ such that
\begin{align*}
\forall e' \in L^{\ell}~ \exists g \in G : d_{e'g} \leq 2\cdot d_{e'e}.
\end{align*}
\end{lemma}
\begin{proof}
  The set $G$ is found as follows (see Figure~\ref{fig:guards1}):
  \begin{inparaenum}[(i)]
  \item Partition the plane into six sectors centered at $r_e$, each
    with an angle of $60^{\circ}$. Denote these sectors by
    $\sector(j)$, where $j \in \{1,\ldots,6\}$.
  \item For every $\sector(j)$, let $e_j \in L^{\ell}$ denote a link
    such that the transmitter $s_{e_j}$ is closest to $r_e$ among the
    transmitters in $\sector(j)$.
  \item Let $g_j$ denote a link in $L^{\ell}$ such that $r_{g_j}$ is
    closest to $s_{e_j}$ (If $\sector(j)$ lacks
    transmitters, then $g_j$ is not defined).
  \end{inparaenum}
  Let $G\eqdf\{r_{g_j}\}_{j=1}^6$ denote the set of guards.

\begin{figure}[H]
  \centering
    \includegraphics[width=0.7\textwidth]{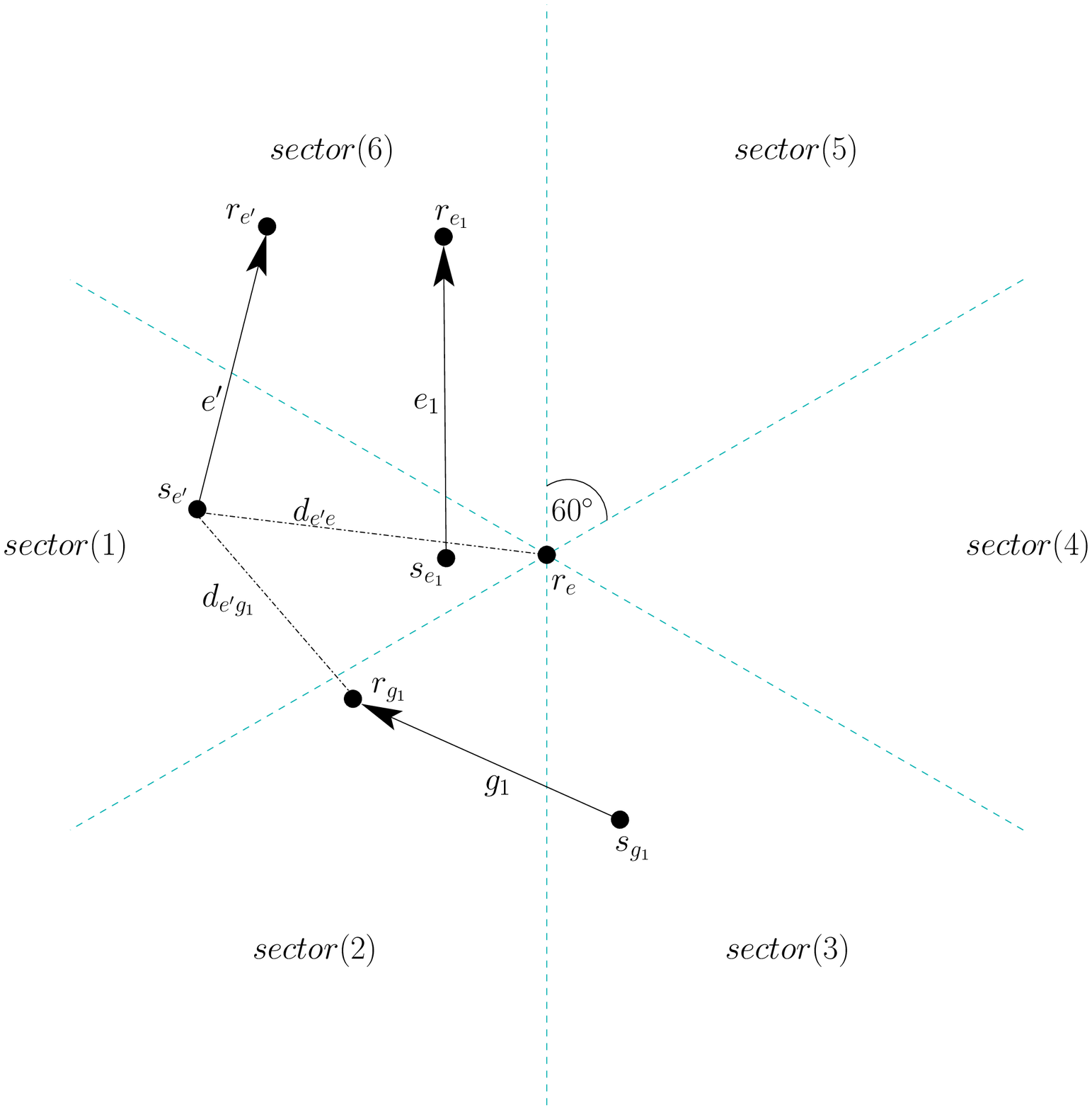}
    \caption{
    A depiction of the proof of Lemma~\ref{lem:guards}.
    %Nodes are depicted by black circles. Links are depicted by arrows. An arrow's tail is located at the transmitter node of the link, while its head is located at the receiver node of the link. The distances $d_{e'e}$ and $d_{e'g_1}$ are depicted by dashed lines. The plane is partitioned into six sectors, depicted by dashed gray lines. The transmitter $s_{e_1}$ is the closest to $r_e$, among the transmitters in \textit{sector}$(1)$. The node $r_{g_1}$ is the closest node to $s_{e_1}$.
    }
  \label{fig:guards1}
\end{figure}

We first consider the case that $e' \in L^{\ell}$ is also a guard ($e'=g_j$).
In this case choose $g=e'$, and $d_{e'g}=d_{e'}$. But $d_{e'} \leq d_{e'e}$ since $e' \in L^{\ell}$, as required.
We now consider the case that $e' \in L^{\ell} \setminus G$.
  Given $e'\in L^{\ell} \setminus G$, let $j$ denote the sector that
  contains $s_{e'}$.  We claim that $d_{e'g_j} \leq 2\cdot d_{e'e}$.
  Consider first $e'=e_j$ (i.e., $s_{e_j}$ is a closest sender to
  $r_e$ in $\sector(j)$).  Since $r_{g_j}$ is a closest receiver to
  $s_{e_j}$, we have $d_{e_jg_j} \leq d_{e_j}$.  Since $e_j\in
  L^{\ell}$, we have $d_{e_j} \leq d_{e_je}$.  Thus, $d_{e_jg_j}\leq
  d_{e_je}$, as required.

  Consider now a link $e' \neq e_j$.  The following
  inequalities hold:
  \begin{align}
\label{eq:1}    d_{e'e} &\geq d_{e_je}, \text{ ($s_{e_j}$ is a closest sender to $r_e$)}\\
\label{eq:2}    d_{e'g_j} &\leq d_{s_{e'}s_{e_j}}+ d_{e_jg_j}, \text{ (triangle ineq. in $\triangle s_{e'}s_{e_j}r_{g_j}$)}\\
\label{eq:3}    d_{e_jg_j} &\leq d_{e_je}, \text{ (already proved for $e_j$)}\\
\label{eq:4}   d_{s_{e'}s_{e_j}} &\leq d_{e'e}. \text{ (proved below)}.
  \end{align}

We now prove Eq.~\ref{eq:4} (see Figure~\ref{fig:guards2}).  Let $s^*$ denote the point along the
segment from $r_e$ to $s_{e'}$ such that $d_{s^*r_e}=d_{e_je}$.  The
triangle $\triangle r_es_{e_j}s^*$ is an isosceles triangle.  Since
$\angle s_{e_j} r_e s^* \leq 60^\circ$, it follows that the base
angle $\angle r_e s_{e_j} s^* \geq
60^\circ$.  Hence, $\angle r_e s_{e_j} s_{e'} \geq \angle r_e s_{e_j} s^* \geq  60^\circ$.  Since $\angle s_{e_j} r_e s_{e'} \leq 60^\circ$,
it follows that $d_{s_{e'},s_{e_j}} \leq d_{e'e}$, as required.

\begin{figure}[H]
  \centering
    \includegraphics[width=0.55\textwidth]{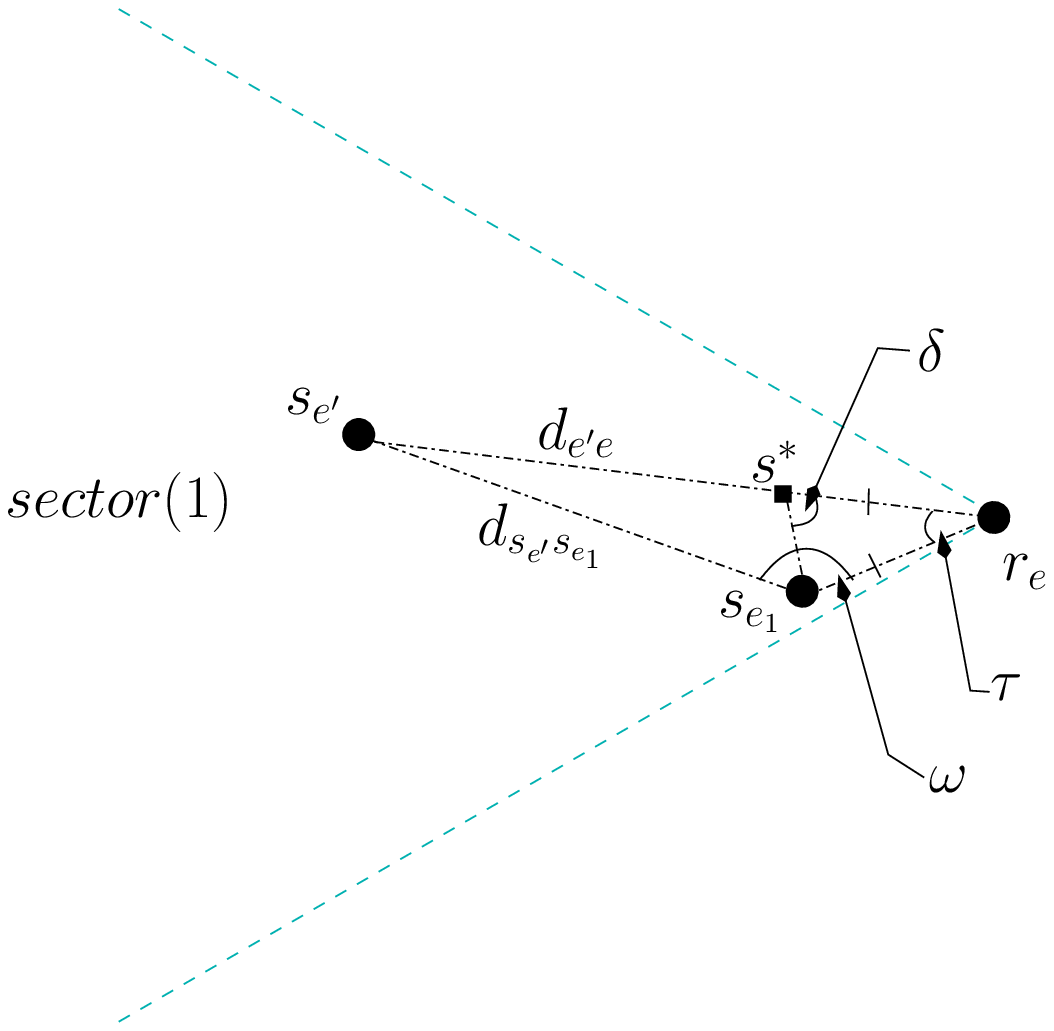}
    \caption{ The
triangle $\triangle r_es_{e_j}s^*$ is an isosceles triangle. The angle $\tau \leq  60^\circ$. The angle $\delta \geq  60^\circ$. The angle $\omega \geq \delta \geq  60^\circ$.}
  \label{fig:guards2}
\end{figure}

To complete the proof that $d_{e'g_j} \leq 2\cdot d_{e'e}$, observe that
\begin{align*}
  d_{e'g_j} \stackrel{\text{eq.~\ref{eq:2}}}{\leq} d_{s_{e'}s_{e_j}}+ d_{e_jg_j}
\stackrel{\text{eqs.~\ref{eq:3},\ref{eq:4}}}{\leq} d_{e'e} + d_{e_je}
\stackrel{\text{eq.~\ref{eq:1}}}{\leq} 2\cdot d_{e'e}.
\end{align*}
\end{proof}

\subsection{Necessary Conditions}
Recall that Let $L \subseteq \LL$ is an \sinr-feasible set of links that
belong to same bucket $B_i$.
Let $e\in B_i$ denote an arbitrary link (not
necessarily in $L$).
\begin{lemma}\label{lemma:in degree}
  $$\sum_{e' \in L ^{\ell}} {\bar a}_{e'}(e) = O(1).$$
\end{lemma}
\begin{proof}
  By Lemma~\ref{lem:guards}, we find a set of ``guards'' $G \subseteq
  L^{\ell}$, such that:
    %\begin{inparaenum}[(i)]
    \begin{enumerate}[(i)]
        \item \label{item: g6} $|G| \leq 6$,
        \item \label{item:gsidt}$\forall e' \in L^{\ell}~~ \exists g \in G : d_{e'g} \leq 2\cdot d_{e'e}.$
    \end{enumerate}
    %\end{inparaenum}
First, let us bound $\sum_{e' \in L ^{\ell} \setminus G} {\hat a}_{e'}(e)$,

\begin{align}
   \sum_{e' \in L ^{\ell} \setminus G} {\hat a}_{e'}(e) & <    \sum_{e' \in L ^{\ell} \setminus G} 2 \cdot \left(\frac{d_{e'}}{d_{e'e}}\right)^{\alpha} \nonumber \\
  &\leq 2^{\alpha+1} \cdot \sum_{e' \in L ^{\ell}\setminus G} \sum_{g \in G}\left(\frac{d_{e'}}{d_{e'g}}\right)^{\alpha} \nonumber \\
  &\leq 2^{\alpha+2} \cdot \sum_{g \in G} \hat{a}_{L^{\ell}} (g) \label{eq:ll1}\:,
\end{align}
where the first line follows from Proposition~\ref{prop:aff}. The second line follows from
Lemma~\ref{lem:guards}. The third line, again, follows from Proposition~\ref{prop:aff}.

Since ${\bar a}_{e'}(e) \leq 1$, we obtain
\begin{eqnarray}
    \sum_{e' \in L ^{\ell}} {\bar a}_{e'}(e) &  \leq    &   \sum_{e' \in L^{\ell} \setminus G}  {\bar a}_{e'}(e) + |G| \label{eq:ll2}\:,
%        &   =    &   \sum_{e' \in L^{\ell} \setminus G}  \min\{1,a_{e'}(e)\} + |G|
\end{eqnarray}
%The third line follows from the definition of ${\bar a}_{e'}(e)$.

Hence,
\begin{eqnarray*}
     \sum_{e' \in L ^{\ell}} {\bar a}_{e'}(e) &\leq & \sum_{e' \in L^{\ell} \setminus G} a_{e'}(e) + |G|\\
                & = & \sum_{e' \in L^{\ell} \setminus G} \gamma_e \cdot {\hat a}_{e'}(e) + |G|\\
                & \leq & \gamma_e \cdot 2^{\alpha+2} \cdot \sum_{g \in G} \hat{a}_{L^{\ell}} (g) + |G|\\
                & \leq &   |G| \cdot\left(\frac {\gamma_e \cdot 2^{\alpha+2}}{\min_{g\in G}\gamma_g} + 1 \right) \\
                & \leq &   6 \left(\frac {(1+\eps)\cdot 2^{\alpha+2}}{\eps} + 1 \right)\:,
\end{eqnarray*}
where the first line follows from Equation~\ref{eq:ll2} and the fact that ${\bar a}_{e'}(e) \leq  a_{e'}(e)$.
The second line follows from the fact that $\gamma_e \cdot {\hat a}_{e'}(e) = a_{e'}(e)$.
The third line follows from Equation~\ref{eq:ll1}.
The fourth line follows since $L^{\ell}$ is \sinr-feasible, that is,  $a_{L^{\ell}} (g)\leq 1$ and $\hat a_{L^{\ell}} (g)\leq 1/\gamma_g$, for every $g \in G$. % \subseteq L ^{\ell}$.
The last line follows from Proposition~\ref{prop:gamma}, Lemma~\ref{lem:guards}, and $|G| \leq 6$.
Since, $\alpha$ and $\eps$ are constants, the lemma follows.
\end{proof}
\begin{lemma}\label{lemma:out degree}
  $$\sum_{e' \in L^g } {\bar a}_{e'}(e) = O(1).$$
\end{lemma}
\begin{proof}
Pick $e^*$ to be a shortest link in $L^{g}$.
It follows from Proposition~\ref{prop:aff} and the triangle inequality (see Figure~\ref{fig:triangle}) that
\[
    \forall e' \in L^{g} \setminus \{e^*\}:
    \hat a_{e'}(e^*) > \frac 12 \cdot \left( \frac{d_{e'}}{d_{e' e^*}}\right)^\alpha \geq
    \frac 12 \cdot \left( \frac{d_{e'}}{d_{e'e} + d_{e^*e} +d_{e^*}}\right)^\alpha\:.
\]
\begin{figure}[H]
  \centering
    \includegraphics[width=0.5\textwidth]{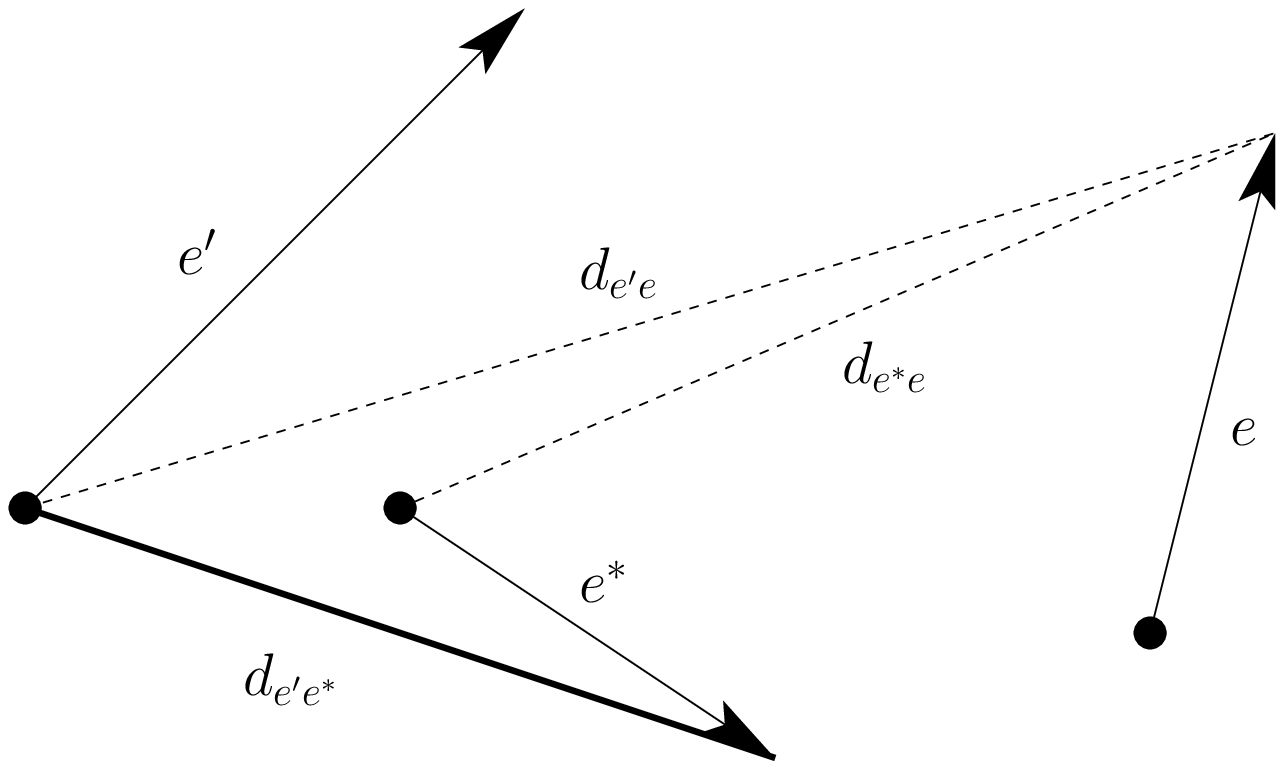}
    \caption{The distance $d_{e'e^*}$ is depicted by a bold segment. We bound $d_{e'e^*}$
    by applying the triangle inequality, that is the dashed segments and the length of link $e^*$, $d_{e^*}$.}
  \label{fig:triangle}
\end{figure}

Since $e',e^* \in L^{g}$, it follows that $d_{e'} > d_{e'e}$ and  $d_{e^*} > d_{e^*e}$.
Since $d_{e'} \geq d_e^* $ it follows that
\[
    \hat a_{e'}(e^*) > \frac 12 \cdot \left( \frac{d_{e'}}{3\cdot d_{e'}}\right)^\alpha > \frac 12 \cdot \frac{1}{3^\alpha}.
\]
Since $a_{L^g}(e^*)= \gamma_{e^*} \cdot {\hat a}_{L^g}(e^*)$, it follows:
$$a_{L^g}(e^*)= \gamma_{e^*}\cdot {\hat a}_{L^g}(e^*) > \frac 12 \cdot \frac{1}{3^\alpha}\cdot \gamma_{e^*} \cdot (|L^{g}| -1)\:.$$
Since $L^{g}$ is \sinr-feasible, it follows that $a_{L^g}(e^*) \leq 1$. Hence,
\begin{eqnarray*}
    \frac 12 \cdot \frac 1{3^\alpha} \cdot \gamma_{e^*}\cdot (|L^{g}| -1) & < & 1 \Rightarrow \\
    |L^{g}| & < & 2 \cdot 3^\alpha/\gamma_{e^*}+1\:.
\end{eqnarray*}
Proposition~\ref{prop:gamma} implies that $\frac {1}{\gamma_{e^*}} = O(1)$.
Since $\alpha$ is a constant, it follows that $|L^{g}| = O(1)$.
Since $\sum_{e' \in L^g } {\bar a}_{e'}(e) \leq |L^{g}|$, the lemma follows.
\end{proof}

\noindent
Lemmas~\ref{lemma:in degree} and~\ref{lemma:out degree} imply the following theorem.
\begin{theorem}\label{thm:in}
  Let $L$ denote an SINR-feasible set of links.  If $L \subseteq
  B_i$, then
     \[
\forall e\in  B_i:~~~\sum_{\{e' \in L : d_{e'} \geq d_{e}\}}\bar a_{e'}(e) ~\leq~ {\bar a}_L(e) + \bar a_{e}(e)= O(1).
     \]
\end{theorem}

\noindent
The following theorem follows from ~\cite[Thm 1]{K10}.
The proof of the following theorem is in Appendix~\ref{sec:proofs}.
\begin{theorem} \label{thm:out}
    Let $L$ denote an SINR-feasible set of links. If $L \subseteq
  B_i$, then
     \[
\forall e\in B_i:~~~           \sum_{\{e' \in L : d_{e'} \geq d_{e}\}}\bar a _{e}(e') = O(1).
     \]
\end{theorem}
%\begin{proof}
%    Theorem 1 in ~\cite{K10} implies that
%    $$\sum_{\{e' \in L : d_{e'} \geq d_{e}\}}\min\left\{1,\left(\frac{d_e}{d_{ee'}}\right)^{\alpha}\right\}
%        +\sum_{\{e' \in L : d_{e'} \geq d_{e}\}}
%    \min \left\{1,\left ( \frac{d_e}{d_{s_{e'}r_e}}\right)^{\alpha}\right\} = O(1).$$
%    It follows that,
%    \begin{eqnarray*}
%        O(1) & = & \sum_{\{e' \in L : d_{e'} \geq d_{e}\}}\min\left\{1,\left(\frac{d_e}{d_{ee'}}\right)^{\alpha}\right\} \\
%        & \geq & \sum_{\{e' \in L : d_{e'} \geq d_{e}\}}\min\left\{1,\frac 12 \cdot{\hat a}_{e}(e')\right\} \\
%        & = & \sum_{\{e' \in L : d_{e'} \geq d_{e}\}}\min\left\{1,\frac {1}{2\cdot \gamma_{e'}} \cdot a_{e}(e')\right\} \\
%        & \geq & \sum_{\{e' \in L : d_{e'} \geq d_{e}\}}\min\left\{1,\frac {\eps}{2\cdot (1+\eps) \cdot \beta} \cdot a_{e}(e')\right\}\:,
%    \end{eqnarray*}
%    where the second line follows from Proposition~\ref{prop:aff}. The third line follows from the definition of $a_e(e')$. The last line follows from Proposition~\ref{prop:gamma}.
%    The theorem follows, since $\frac {\eps}{2\cdot (1+\eps) \cdot \beta} = O(1)$ and since $\bar a_{e'}(e)\eqdf \min\{1,a_{e'}(e)\}$.
%\end{proof}

\section{LP Relaxation}\label{sec:LP}
In this section we formulate the linear program for the \MAX\ and \MAXMIN\ problems with arbitrary power assignments.  The linear
program formulation that we use for computing the multi-commodity flow
$f$ is as follows.

\begin{align}
\LPmax : \XTH^* =   \text{maximize } &|f| \text{ subject to} \nonumber\\
f &\in \FF \label{const:ff}\\
\forall i~~ \forall e & \in B_i ~~~ f(e)+\sum_{\{e'\in B_i : d_{e'}\geq d_e\}} (\bar{a}_{e'}(e) + \bar{a}_e(e')) \cdot f(e') \leq 1
\label{eq:int}
\end{align}

\begin{align}
\LPmin :\XNTH^*= & \text{maximize } \rho \text{ subject to} \nonumber\\
f&\in \FF_\rho \label{const:ffrho}\\
\forall i ~~\forall e & \in B_i ~~~~ f(e)+\sum_{\{e'\in B_i : d_{e'}\geq d_e\}} (\bar{a}_{e'}(e) + \bar{a}_e(e')) \cdot f(e') \leq 1\label{const:nes2}
\end{align}

Recall that  $\FF$ denotes the polytope of all multi-commodity flows $f=(f_1,\ldots,f_k)$
such that $|f_i|\leq b_i$, for every $i$.
Also recall that $\FF_\rho\subseteq \FF$ for $\rho>0$ denotes the polytope of
all multi-commodity flows such that $|f_i|/b_i\geq \rho$.
Constraints~\ref{const:ff},~\ref{const:ffrho} in $\LPmax$\ and $\LPmin$\ respectively require that the
$f$ is a feasible multi-commodity flow with respect to $\FF$ and $\FF_\rho$.

Constraints ~\ref{eq:int},~\ref{const:nes2} in $\LPmax$ and $\LPmin$ respectively require that for every bucket $B_i$ and for every link $e \in B_i$ the amount of flow $f(e)$ plus the amount of the weighted symmetric interferences is bounded by one. Note that this symmetric interference constraint is with respect to links that are longer than $e$.

The objective function of $\LPmax$ is to maximize the total flow $|f|$.
The objective function of $\LPmin$ is to maximize $\rho$, such that $f\in \FF_\rho$. Namely, maximize $\min_{i=1\ldots k} |f_i|/b_i$.

We prove on Section~\ref{sec:analysis} that the linear programs $\LPmax$ and $\LPmin$ are relaxations of the \MAX\ and \MAXMIN\ problems.

\section{Algorithm}\label{sec:alg}
\subsection{Algorithm description}

For simplicity, we assume in this section that all the links are in the same bucket, that is $\LL \subseteq B_i$ for some $i$.
In Section~\ref{sec:arbitraty} we show how to handle arbitrary power assignment.
In Section~\ref{sec:limited} we extend the algorithm so that it assigns limited powers.

\paragraph{Algorithm overview.}
We overview the algorithm for the \MAX\ problem.
Assume for simplicity that, $\LL \subseteq B_i$ for some $i$.
First, the optimal flow $f^*$ is obtained by solving the linear program $\LPmax$.
We need to find an \sinr-feasible schedule that supports a fraction of $f^*$.
Second, we color the links using greedy multi-coloring. This coloring induces a preliminary schedule, in which every color class is ``almost'' \sinr-feasible. This preliminary schedule is almost \sinr-feasible since in every color class and every link $e$, the affectance of links that are longer than $e$ on $e$ is at most 1. However, the affectance of shorter links on $e$ may be still unbounded.
Finally, we refine this schedule in order to obtain an \sinr-feasible schedule.
Note that the returned \sinr-feasible  schedule supports a fraction of the flow $f^*$. We show in Section~\ref{sec:analysis} that this fraction is at least $\Omega(1/\log n)$.

\paragraph{Algorithm description.}
The algorithm for the \MAX\ problem proceeds as follows.
\begin{enumerate}
\item \label{item:alg1}Solve the linear program $\LPmax$. Let $f^*$ denote the optimal solution.
\item \label{item:alg15} Remove flow paths that traverse edges with
  $f^*(e) < 1/(2nm)$.  Let $\hat f$ denote the remaining flow.
\item \label{item:alg2} Set $T=2nm$. Apply the greedy multi-coloring
  algorithm $\gcolor$ (see Section~\ref{sec:coloring}) on the input
  $((\LL,\LL^2),\hat{f},d,w,T)$, where the pair $(\LL,\LL^2)$ is a
  complete graph whose set of vertices is $\LL$, for every link in $e
  \in \LL$, $d(e) = d_e$, and $w(e,e')\eqdf \bar a_e(e')+ \bar
  a_{e'}(e)$ is a weight function over pair of links in $\LL$. Let
  $\pi:\LL\rightarrow 2^{\{0,\ldots T-1\}}$ denote the computed
  multi-coloring.
\item \label{item:alg3} Apply procedure $\disperse$ to each color class $(\pi^{-1}(t))$, where $t\in \{0,\ldots T-1\}$. Let $\{L_{t,i}\}_{i=1}^{\ell(t)}$ denote the
    dispersed subsets.
  \item \label{item:alg4}Return the schedule $\{L_{t,i}\}_{t=0..T-1,i=1..\ell(t)}$
    and the flow $f=(f_1,\ldots,f_k)$, where $f=\hat{f}/(2\cdot \ell(t))$.
\end{enumerate}

Clearly steps~\ref{item:alg1} and~\ref{item:alg4} are polynomial.
In Section~\ref{sec:coloring} we show that step~\ref{item:alg2} is polynomial.
In Section~\ref{sec:disperese} we show that \disperse\ is polynomial.
Therefore, the running time of the algorithm is polynomial.

\begin{rem}
    The following changes are needed in order to obtain an algorithm for the \MAXMIN\ problem:
    \begin{inparaenum}[(i)]
        \item In Item~\ref{item:alg1} solve the linear program $\LPmin$,
        \item in Item~\ref{item:alg15} remove flow paths that traverse edges
            with $f^*(e) < 1/(2n^2km)$,
        \item in Item~\ref{item:alg2} set $T=2n^2km$.
    \end{inparaenum}
\end{rem}

\subsection{Removing Minuscule Flow Paths}\label{sec:remove}

The greedy multi-coloring algorithm cannot support flows $f^*(e) <
1/(2nm)$.  We mitigate this problem simply by peeling off flow
paths that traverse edges with a flow smaller than $1/(2nm)$.  The
formal description of this procedure is as follows.
\begin{inparaenum}[(1)]
\item Initialize $\hat{f} \gets f$.
\item While there exists an edge $e$ with $\hat{f}(e) < 1/(2nm)$,
  remove flow from $\hat{f}$ until $\hat{f}(e)=0$.  This is done by
  computing flow paths for the flow that traverses $e$, and zeroing
  the flow along these paths.
\end{inparaenum}

\subsection{Greedy Multi-Coloring}\label{sec:coloring}
Let $G=(V,E)$ denote an undirected graph with edge weights $w : E
\rightarrow [0,1]$ and node demands $x : V \rightarrow [0,1]$.  Assume
an ordering of the nodes induced by distinct node lengths $d(v)$. For a set $V'\subset
  V$, let $w(V',u)\eqdf \sum_{v\in V'} w(v,u)$.
Assume that
\begin{align}\label{eq:color}
  \forall u \in V : x(u) + \sum_{\{v \in V : d(v) > d(u)\}}w(v,u)\cdot
  x(v) \leq 1\:.
\end{align}

Indeed, Constraints~\ref{eq:int},~\ref{const:nes2} in $\LPmax$ and $\LPmin$,  respectively, imply that the input to the greedy coloring algorithm satisfies the assumption in Equation~\ref{eq:color}.

\begin{lemma}[Greedy Coloring Lemma] \label{lemma:color}
%  Let $T \eqdf \lceil \frac {1}{\min_{v\in V} x(v)} \rceil$.
For every integer $T$, there is
  multi-coloring $\pi : V \rightarrow 2^{\{0,\ldots,T-1\}}$, such that
  \begin{enumerate}
  \item $\forall c\in \{0,\ldots,T-1\} ~~\forall u \in \pi^{-1}(c):
    \sum_{\{v \in V: d(v) > d(u)\}} w(v,u)\leq 1$, \label{item:color1}
  \item $\forall u \in V : |\pi(u)| \geq \lfloor x(u) \cdot T
    \rfloor$.
  \end{enumerate}
\end{lemma}

\begin{algorithm}
\caption{$\gcolor((V,E),x,d,w,T)$ - greedy multi-coloring of $V$. }
\label{alg:color}
  \begin{enumerate}
  \item Scan the vertices in descending $d(v)$ length order, let $u$ denote the current node.
    \begin{enumerate}
    \item $\Cbad \gets \{ c\in \Tzo : w(\pi^{-1}(c),u)>1 \}$.
    \item If $|\Cbad| > T-\lfloor x(u) \cdot T \rfloor$, then return
      ``FAIL''.
    \item $\pi(u) \gets $ first $\lfloor x(u) \cdot T \rfloor$ colors
      in $\Tzo\setminus \Cbad$.
    \end{enumerate}
  \item Return $(\pi)$.
  \end{enumerate}
\end{algorithm}

The running time of Algorithm~\ref{alg:color} is at most $O(T^2\cdot |V|\cdot|E|)$.
Since $T,|E|$ and $|V|$ are polynomial, it follows that the running time is polynomial.

\begin{proof}
  We apply a ``first-fit'' greedy multi-coloring listed in
  Algorithm~\ref{alg:color}.  We now prove that this algorithm
  succeeds.

  Let $b(u) \eqdf \lfloor x(u) \cdot T \rfloor$.  Assume, for the sake
  of contradiction that, $|C^{\text{bad}}_u| > T-b(u)$, hence,
    \begin{eqnarray}
      T-b(u)+1&\leq& |C^{\text{bad}}_u|\nonumber\\
     & \leq & \sum_{c \in C^{\text{bad}}_u} w(\pi^{-1}(c),u) \nonumber\\
      &\leq  & \sum_{\{v:d(u)<d(v)\}} |\pi(v)|\cdot w(v,u) \nonumber\\
      & =    & \sum_{\{v:d(u)<d(v)\}} b(v) \cdot w(v,u)\:.\label{eq:cond1}
    \end{eqnarray}
    The third line follows from the fact that vertices are scanned in
    a descending length order, and by a
    rearrangement of the summation order.
 By adding $b(u)$ to both sides, we obtain:
    \begin{align}
        T+1 \leq \lfloor x(u) \cdot T \rfloor+\sum_{\{v:d(u)<d(v)\}} \lfloor x(v) \cdot T \rfloor \cdot w(v,u). \label{eq:cond2}
    \end{align}
    We divide Eq.~\ref{eq:cond2} by $T$ to obtain a contradiction to Eq.~\ref{eq:color}, as required.
We conclude, that the greedy coloring succeeds, and the lemma follows.
\end{proof}

\subsection{The dispersion procedure \disperse}\label{sec:disperese}
The input to the dispersion procedure \disperse\ consists of a set $L\subseteq
\LL$ of links that are assigned the same color by the multi-coloring
procedure (see Algorithm~\ref{alg:color} in Section~\ref{sec:coloring}).
This implies that
\begin{align}\label{eq:color1}
  \forall e\in L: \sum_{\{e'\in L \setminus \{e\}: d_{e'}\geq d_e\}} (\bar a_e(e')+
  \bar a_{e'}(e))\leq 1.
\end{align}

The dispersion procedure works in two phases.  In the first phase, $L$
is partitioned into $\overline{1/3}$-\barsignal\ sets $\{L_i\}_i$. In the
second phase, each subset $L_i$ is further partitioned into $\overline{7/6}$-\barsignal\ sets $\{L_{i}\}_{i=1}^{\ell(t)}$.
Recall that a set of links $L_i$ is \sinr-feasible if $L_i$ is a $\overline{(1+\eps)}$-\barsignal\ for some $\eps > 0$.
Since every set in  $\{L_{i}\}_{i=1}^{\ell(t)}$ is $\overline{(7/6)}$-\barsignal, it follows that every set in $\{L_{i}\}_{i=1}^{\ell(t)}$ is \sinr-feasible.

In Algorithm~\ref{alg:disperse3}, we list the first phase of the
dispersion procedure.  Note that if a $\overline{1/3}$-\barsignal\ set $J^i$ is
always found in Line~\ref{line:J}, then $L$ is dispersed into at most
$\log_2 |L|$ subsets. In Lemma~\ref{lemma:disperse} we prove that this
is indeed possible.

\begin{algorithm}
  \caption{$\frac 13$-\disperse$(L)$ : partition $L\subseteq \LL$ into $O(\log n)$
    $\overline{1/3}$-\barsignal\ sets.}
\label{alg:disperse3}
      \begin{enumerate}
      \item $i\gets 0$ and $L^0 \gets L$.
      \item \label{item:disalg2}while $L^i\neq \emptyset$ do
        \begin{enumerate}
        \item find a $\overline{1/3}$-\barsignal\ set $J^i\subseteq L^i$ such that
          $|J^i|\geq |L^i|/2$.
\label{line:J}
        \item $L^{i+1}\gets L^i\setminus J^i$ and $i\gets i+1$.
        \end{enumerate}
      \end{enumerate}
\end{algorithm}

The second phase follows~\cite[Thm 1]{HW}. This phase is implemented by
two first-fit bin packing procedures. In the first procedure, open $7$
bins, scan the links in some order and assign each link to the first
bin in which its affectance is at most $3/7$.  In the second procedure,
partition each bin into $7$ sub-bins.  Scan the links in the reverse
order, and again, assign each link to the first bin in which
its affectance is at most $3/7$.

Proposition~\ref{prop:disperse} implies that step~\ref{item:disalg2} in Algorithm~\ref{alg:disperse3} terminates after $O(\log m)$ iterations. Each of these iterations is polynomial.
The second phase of the \disperse\ algorithms is clearly polynomial.
Therefore, the running time of the \disperse\ algorithm is polynomial.

\section{Algorithm Analysis}\label{sec:analysis}

In this section we analyze the algorithm presented in
Section~\ref{sec:alg}.  Recall that it is assumed that all the links
are in the same bucket, that is $\LL \subseteq B_i$ for some $i$.
First, we prove that the linear program $\LPmax$ is a fractional
relaxation of the \MAX\ problem.  We
then show that the greedy coloring computes a schedule that supports
the flow given by the LP. Unfortunately, this schedule is not an
SINR-feasible schedule. We then prove that the refinement procedure
(Step~\ref{item:alg3} of the algorithm) generates an SINR-feasible schedule with an
$O(\log n)$ increase in the approximation ratio.

\medskip
\noindent
Let $f^*$ denote an optimal solution of the linear program $\LPmax$,
i.e., $\XTH^*=|f^*|$.  The following lemma shows that the linear
program $\LPmax$ is a relaxation of the \MAX\ problem.
\begin{lemma}\label{lemma:relax}
  There exists a constant $\lambda \geq 1$ such that, if
  $S=\{L_t\}_{t=0}^{T-1}$ is an \sinr-feasible schedule that supports
  a multi-commodity flow $f$, then $f/\lambda$ is a feasible solution
  of the linear program $\LPmax$. Hence, $F^*\geq |f|/\lambda$.
\end{lemma}
\begin{proof}
  Clearly $f/\lambda \in \FF$. Thus, we only need to prove that
  $f/\lambda$ satisfies the constraint in Eq.~\ref{eq:int}.  Consider
  an \sinr-feasible set $L_t$ and an arbitrary link $e$. By,
  Theorems~\ref{thm:in} and~\ref{thm:out}, $$\sum_{\{e'\in L_t :
    d_{e'}\geq d_e\}} (\bar{a}_{e'}(e) + \bar{a}_e(e')) \leq O(1).$$  It
  follows that
\begin{align}\label{eq:avg}
  \frac{1}{T}\cdot
\sum_{t=0}^{T-1} \sum_{\{e'\in L_t : d_{e'}\geq d_e\}}
  (\bar{a}_{e'}(e) + \bar{a}_e(e')) \leq O(1).
\end{align}
Since $f(e')\leq \frac 1T \cdot |\{t: e'\in L_{t}\}|$,
We conclude that
\begin{align}\label{eq:f}
    \frac{1}{T}\cdot
\sum_{t=0}^{T-1} \sum_{\{e'\in L_t : d_{e'}\geq d_e\}}
  (\bar{a}_{e'}(e) + \bar{a}_e(e'))
&\geq \sum_{\{e'\in {\cal L} : d_{e'}\geq d_e\}}
  (\bar{a}_{e'}(e) + \bar{a}_e(e')) \cdot f(e').
\end{align}
Since $f(e)\leq 1$, we conclude from Eqs.~\ref{eq:avg} and~\ref{eq:f} that
\begin{align}\label{eq:last}
  f(e)+ \sum_{\{e'\in {\cal L} : d_{e'}\geq d_e\}} (\bar{a}_{e'}(e) +
  \bar{a}_e(e')) \cdot f(e')\leq O(1).
\end{align}
Let $\lambda>0$ denote a constant that bounds the left-hand side in
Eq.~\ref{eq:last}.  Then, $f/\lambda$ satisfies the constraints in
Eq.~\ref{eq:int}, as required, and the lemma follows.
\end{proof}

Analogously, one could prove also that
the linear program $\LPmin$ is a relaxation
of the \MAXMIN\ problem.
\begin{lemma}
  Suppose $S=\{L_t\}_{t=0}^{T-1}$ is an \sinr-feasible schedule that
  supports a multi-commodity flow $f$.  If $\rho\eqdf \min_{i=1 \ldots k} |f_i|/b_i$,
  $R^* \geq \rho/\lambda$, for the same constant $\lambda \geq 1$ in
  Lemma~\ref{lemma:relax}.
\end{lemma}

\noindent
The following proposition gives a lower bound on the optimal throughput.
\begin{proposition} \label{prop:thh1}
    $F^* \geq \frac 1n$ and $R^*\geq \frac{1}{n^2k}$.
\end{proposition}
\begin{proof}
  Without loss of generality, the source and destination of each
  request are connected.  Pick a request $R_i$ and a path $p_i$ from
  ${\hat s}_i$ to ${\hat t}_i$.  Consider the schedule that schedules the links of
  $p_i$ in a round-robin fashion. Clearly, this schedule supports a
  flow $f=1/|p|$ from ${\hat s}_i$ to ${\hat t}_i$ along $p$, where $|p|$ denotes
  the length of $p$. This implies that $F^*\geq 1/n$, as required.
  The second part of the proposition is proved by concatenating $k$
  schedules, one schedule per request. The concatenated schedule
  supports a flow $f=(f_1,\ldots,f_k)$, where $f_i=1/(nk)$ along the
  path $p_i$. Since $b_i\leq n$, it follows that $|f_i|/b_i \geq
  1/(n^2k)$, and the proposition follows.
\end{proof}

\begin{proposition}\label{prop:miniscale}
$|{\hat f}| \geq\ F^*/2$
\end{proposition}
\begin{proof}
  Let us denote by $g$ the total flow that was removed in step~\ref{item:alg15}.
  The contribution to the flow amount $|g|$ due to edges with small flow is less
  than $1/(2nm)$.
  Since there are $m$
  edges, it follows that $|g| \leq 1/(2n)$.  By Prop.~\ref{prop:thh1} we have $F^* \geq
  \frac 1n$, and the proposition follows.
\end{proof}

\noindent
For the case of $\LPmin$, one can show a similar result, that is $|{\hat f}| \geq\ R^*/2$.

\begin{proposition}\label{prop:T}
  If  $~T\geq 2nm$, then the greedy multi-coloring algorithm computes a
  multi-coloring $\pi$ that induces a schedule that supports ${\hat f}/2$.
\end{proposition}
\begin{proof}
   Recall that a schedule $S=\{L_t\}_{t=0}^{T-1}$ induced by a multi-coloring 
   $\pi:\LL \rightarrow 2^{\{0,\ldots,T-1\}}$ is defined by 
   $\forall t:~L_t=\pi^{-1}(t)$, where $\pi^{-1}(t)\eqdf\{e: t\in \pi(e)\}$.
  Also recall that a schedule $S$ supports $\hat f$ if 
  $\forall e\in \LL:~~~  T \cdot {\hat f}(e) \leq \left|\{t\in \{0,\ldots,T-1\} : e\in L_t\}\right|$.
  Lemma~\ref{lemma:color} implies that the greedy multi-coloring algorithm 
  (see the listing in Algorithm~\ref{alg:color}) computes multi-coloring 
  $\pi$ such that $\forall e \in \LL : |\pi(e)| \geq \lfloor {\hat f}(e) \cdot T
    \rfloor$.
  Hence, it suffices to prove that $T \cdot \hat{f}(e)/2 \leq \lfloor T \cdot
  \hat{f(e)} \rfloor$, for every edge $e$. 
  Indeed, step~\ref{item:alg15} in the algorithm (see listing in Sec.~\ref{sec:alg}) 
  implies that if $\hat{f}(e)>0$, then $\hat{f}(e)\geq 1/T$. 
  Let us consider the following two cases: (1)~If $\hat{f}(e)\in [1/T,2/T)$, 
  then $T \cdot \hat{f(e)}/2 < 1=\lfloor T \cdot \hat{f(e)} \rfloor$, 
  (2)~if $\hat{f}(e)\geq 2/T$, then $T \cdot \hat{f(e)}/2 \leq T\cdot
  (\hat{f}(e) - 1/T) \leq \lfloor T \cdot \hat{f}(e) \rfloor$, as required.
\end{proof}
\noindent
For the case of $\LPmin$, one can show the same result if $T\geq 2n^2 km$.

\begin{lemma}\label{lemma:disperse}
  If $L\subseteq \LL$ satisfies Eq.~\ref{eq:color1}, then there exists
  a subset $J\subseteq L$ such that:
  \begin{inparaenum}[(i)]
  \item $J$ is a $\overline{1/3}$-\barsignal, and
  \item $|J|\geq |L|/2$.
  \end{inparaenum}
\end{lemma}

\begin{proof}
  Define a square matrix $A$, the rows and columns of which are
  indexed by $L$ as follows: order $L$ in descending length order, so that $e'$
  precedes $e$ if $d_{e'}> d_e$.  Let $A({e,e'}) \eqdf (\bar
  a_e(e')+\bar a_{e'}(e))$ and $A(e,e)=0$. Note that $A$ is symmetric.

  Let $\At$ denote the upper right triangular submatrix of $A$.
  Eq.~\ref{eq:color1} implies that, $$\sum_{\{e': d_{e'} \geq d_e\}}
  A(e',e) \leq 1.$$ Hence, the weight of every column in $\At$ is
  bounded by $1$. This implies that the sum of the entries in $\At$ is
  bounded by $|L|$.  By Markov's Inequality, at most half the rows in
  $\At$ have weight greater than $2$.  Let $J\subseteq L$ denote the
  indexes of the rows in $\At$ whose weight is at most $2$.
Clearly, $|J|\geq |L|/2$.

We claim that, for every $e\in J$, the weight of the column $A^e$ is
at most $3$. Indeed, $\sum_{\{e' :d_{e'}\geq d_e\}} A(e',e) \leq 1$. In
addition, $\sum_{\{e' :d_{e'}< d_e\}} A(e',e)= \sum_{\{e' :d_{e'}< d_e\}}
A(e,e')\leq 2$ since this is the sum of the row indexed $e$ in $\At$.
This implies that ${\bar a}_J(e)\leq 3$, for every $e\in J$, and the lemma follows.
\end{proof}

\begin{proposition}\label{prop:disperse}
  The dispersion procedure partitions every color class $\pi^{-1}(t)$
  into $O(\log m)$ \sinr-feasible sets.
\end{proposition}
\begin{proof}
  Recall that the dispersion procedure \disperse\ consists of two phases.
  The first phase is the $\frac{1}{3}$-\disperse$(\pi^{-1}(t))$ algorithm (see the listing in Algorithm~\ref{alg:disperse3}), 
  and the second phase is implemented by two first-fit packing procedures.
  
  Let us consider the first phase. Note that $L^0 = \pi^{-1}(t)$.
  Since $|L^{i+1}| \leq |L^i|/2$, then  $\frac{1}{3}$-\disperse$(\pi^{-1}(t))$ 
  requires at most $\log_2 |\pi^{-1}(t)|$ iterations. 
  Hence, it partitions $\pi^{-1}(t)$ into at most $\log_2 |\pi^{-1}(t)|$ sets, where each set is a
  $\overline{1/3}$-\barsignal\ set.  
  
  Now, in the second phase each of these sets is partitioned into $49$ subsets.
  The lemma follows.
\end{proof}

\begin{theorem}\label{thm:log}
  If Assumption~\ref{assume:gamma} holds, and all the links are in the same bucket,
  then there exists an $O(\log n)$-approximation algorithm for the \MAX\ and the \MAXMIN\ problems.
\end{theorem}
\begin{proof}
    Let $\opt$ denote the maximum total throughput.  By
  Lemma~\ref{lemma:relax}, $F^*\geq \opt/\lambda= \Omega(\opt)$.  Recall that
  $f^*$ denotes an optimal solution of $\LPmax$.  By
  Prop.~\ref{prop:miniscale} $|\hat{f}| \geq |f^*|/2$, and by
  Prop.~\ref{prop:T}, the multi-coloring $\pi$ supports $\hat{f}/2$.
  By Prop.~\ref{prop:disperse}, the dispersion procedure reduces the
  throughput by a factor of $O(\log m)$.  Since there are no parallel
  edges, $\log m = O(\log n)$.  Thus, the final throughput is
  $|\hat{f}|/O(\log n) = \opt/O(\log n)$, and the theorem follows.
\end{proof}

Since in the linear power assignment all links receive with same power, all the links are in the same bucket. We conclude with the following result for the linear power assignment.
\begin{coro}
  If Assumption~\ref{assume:gamma} holds, then there exists an $O(\log n)$-approximation algorithm for the \MAX\ and the \MAXMIN\ problems in the linear power assignment.
\end{coro}

\section{Given Arbitrary Transmission Powers}\label{sec:arbitraty}
In this section we show how to apply the algorithm presented in
Section~\ref{sec:alg} to the case in which transmission power $P_e$ of
each link $e$ is part of the input. Note that $P_e$ may be arbitrary.

\begin{theorem}\label{thm:given}
  If Assumption~\ref{assume:gamma} holds, then there exists an $O(\log
  n \cdot (\log \Delta + \log \Gamma ))$-approximation algorithm for
  the \MAX\ and the \MAXMIN\ problems when the link transmission
  powers are part of the input.
\end{theorem}
\begin{proof sketch}
    We construct an \sinr-feasible schedule and its supported flow.
    The construction  proceeds as follows: (1)~solve the matching LP, (2)~remove the minuscule flow paths as described in Item~\ref{item:alg15}, (3)~run Items~\ref{item:alg2}-\ref{item:alg4} for every bucket separately, (4)~concatenate the output schedules, to obtain an \sinr-feasible schedule of all the links in $\LL$.
    Step (3) of this construction reduces the flow by a factor of at most  $O(\log n)$.
    Step (4) of this construction reduces the flow by an additional factor of at most the number of nonempty buckets, that is $O(\log \Delta + \log \Gamma )$.
\end{proof sketch}

\section{Limited Powers}\label{sec:selected}\label{sec:limited}
In this section we consider the case in which the algorithm needs to
assign a power $P_e$ to each link. The assigned powers must satisfy
$\pmin \leq P_e \leq \pmax$. To simplify the description, assume that
$\log_2 (\pmax/\pmin)$ is an integer, denoted by $\ell$.

We reduce this problem to the case of given arbitrary powers as
follows.  For each pair of nodes $(u,v)$, define $\ell+1$ parallel links,
where the transmission power of the $i$th copy equals $2^i \cdot
\pmin$.
\begin{theorem}
  Assume that, for every link $e$, $(\pmin/d_e^\alpha)/N \geq
  (1+\eps)\cdot \beta$.  Then, there exists an $O((\log n+\log\log
  \Gamma) \cdot (\log \Delta + \log \Gamma ))$-approximation algorithm
  for the \MAX\ and the \MAXMIN\ problems when the link transmission
  powers are in the range $[\pmin,\pmax]$.
\end{theorem}
\begin{proof sketch}
  Note that the number of links increases by a factor of $O(\log
  \Gamma)$. This implies that the $\log n$ factor increases to $(\log
  n+\log\log \Gamma)$.

  The important observation is that there exists a solution that uses
  the discrete power assignments $2^i\cdot P_e$ and achieves a
  throughput that is a constant fraction of the optimal throughput.
  The theorem follows then from Theorem~\ref{thm:given}.

  The proof of this observation proceeds as follows. Given an optimal
  schedule, refine each time slot so that it is a $p$-signal for
  $p=2$.  This reduces the throughput only by a constant factor (see ~\cite[Thm 1]{HW}).
  Round up each transmission power to the smallest discrete power that satisfies Assumption~\ref{assume:gamma}.
  This increases the affectance by at most a factor of two,
  thus the resulting schedule is \sinr-feasible. Moreover, the
  schedule uses links with powers that satisfy Assumption~\ref{assume:gamma}.
\end{proof sketch}

\subsection*{Acknowledgments}
We thank Nissim Halabi and Moni Shahar for useful conversations.
This project was partially funded by the Israeli ministry of Science and Technology.

\begin{comment}
\begin{verbatim}

online

Different beta
close to linear powers
beta < 1

Euclidian

Integrality issues
\end{verbatim}
\end{comment}

%\nocite{*}
\bibliographystyle{alpha}
\bibliography{icalp11}

\newcommand{\etalchar}[1]{$^{#1}$}
\begin{thebibliography}{GWHW09}

\bibitem[ABL05]{alicherry2005joint}
M.~Alicherry, R.~Bhatia, and L.E. Li.
\newblock {Joint channel assignment and routing for throughput optimization in
  multi-radio wireless mesh networks}.
\newblock In {\em MobiCom}, pages 58--72. ACM, 2005.

\bibitem[Cha09]{ChafekarPhD}
D.R. Chafekar.
\newblock {\em {Capacity Characterization of Multi-Hop Wireless Networks-A
  Cross Layer Approach}}.
\newblock PhD thesis, Virginia Polytechnic Institute and State University,
  2009.

\bibitem[CKM{\etalchar{+}}07]{chafekar2007cross}
D.~Chafekar, VS~Kumar, M.V. Marathe, S.~Parthasarathy, and A.~Srinivasan.
\newblock {Cross-layer latency minimization in wireless networks with SINR
  constraints}.
\newblock In {\em MobiHoc}, pages 110--119. ACM, 2007.

\bibitem[CKM{\etalchar{+}}08]{ChafekarCapacity}
D.~Chafekar, VSA Kumart, M.V. Marathe, S.~Parthasarathy, and A.~Srinivasan.
\newblock {Approximation algorithms for computing capacity of wireless networks
  with SINR constraints}.
\newblock In {\em INFOCOM 2008}, pages 1166--1174, 2008.

\bibitem[FKV10]{fanghanel2010improved}
A.~Fangh{\"a}nel, T.~Kesselheim, and B.~V{\"o}cking.
\newblock {Improved algorithms for latency minimization in wireless networks}.
\newblock {\em Theoretical Computer Science}, 2010.

\bibitem[Gal68]{gallager1968information}
R.G. Gallager.
\newblock {\em {Information theory and reliable communication}}.
\newblock John Wiley \& Sons, Inc. New York, NY, USA, 1968.

\bibitem[GK00]{gupta2000capacity}
P.~Gupta and P.R. Kumar.
\newblock {The capacity of wireless networks}.
\newblock {\em IEEE Transactions on information theory}, 46(2):388--404, 2000.

\bibitem[GOW07]{goussevskaia2007complexity}
O.~Goussevskaia, Y.A. Oswald, and R.~Wattenhofer.
\newblock {Complexity in geometric {SINR} }.
\newblock In {\em MobiHoc}, pages 100--109. ACM, 2007.

\bibitem[GWHW09]{goussevskaia2009capacity}
O.~Goussevskaia, R.~Wattenhofer, M.~Halld{\'o}rsson, and E.~Welzl.
\newblock {Capacity of arbitrary wireless networks}.
\newblock In {\em INFOCOM 2009}, pages 1872--1880, 2009.

\bibitem[Hal09]{halldorsson2009wireless}
M.~Halld{\'o}rsson.
\newblock {Wireless scheduling with power control}.
\newblock {\em ESA 2009}, pages 361--372, 2009.

\bibitem[HM11a]{halldorssonwireless}
M.~Halld{\'o}rsson and P.~Mitra.
\newblock {Wireless Capacity with Oblivious Power in General Metrics}.
\newblock In {\em SODA}, 2011.

\bibitem[HM11b]{halldorsson2011nearly}
M.M. Halldorsson and P.~Mitra.
\newblock Nearly optimal bounds for distributed wireless scheduling in the sinr
  model.
\newblock {\em Arxiv preprint arXiv:1104.5200}, 2011.

\bibitem[HW09]{HW}
M.~Halld{\'o}rsson and R.~Wattenhofer.
\newblock {Wireless Communication is in APX}.
\newblock {\em Automata, Languages and Programming}, pages 525--536, 2009.

\bibitem[Kes11]{K10}
T.~Kesselheim.
\newblock {A constant-factor approximation for wireless capacity maximization
  with power control in the SINR model}.
\newblock In {\em Proceedings of the 22nd ACM-SIAM Symposium on Discrete
  Algorithms (SODA)}, 2011.

\bibitem[KV10]{kesselheim2010distributed}
T.~Kesselheim and B.~V{\"o}cking.
\newblock {Distributed contention resolution in wireless networks}.
\newblock {\em Distributed Computing}, pages 163--178, 2010.

\bibitem[LSS06]{lin2006tutorial}
X.~Lin, N.B. Shroff, and R.~Srikant.
\newblock {A tutorial on cross-layer optimization in wireless networks}.
\newblock {\em Selected Areas in Communications, IEEE Journal on},
  24(8):1452--1463, 2006.

\bibitem[MW06]{moscibroda2006complexity}
T.~Moscibroda and R.~Wattenhofer.
\newblock {The complexity of connectivity in wireless networks}.
\newblock In {\em Proc. of the 25th IEEE INFOCOM}. Citeseer, 2006.

\bibitem[Ton10]{tonoyan2010algorithms}
T.~Tonoyan.
\newblock {Algorithms for Scheduling with Power Control in Wireless Networks}.
\newblock {\em Arxiv preprint arXiv:1010.5493}, 2010.

\bibitem[Wan09]{wan2009multiflows}
P.J. Wan.
\newblock {Multiflows in multihop wireless networks}.
\newblock In {\em MobiHoc}, pages 85--94. ACM, 2009.

\bibitem[WFJ{\etalchar{+}}11]{wanwireless}
P.J. Wan, O.~Frieder, X.~Jia, F.~Yao, X.~Xu, and S.~Tang.
\newblock Wireless link scheduling under physical interference model.
\newblock 2011.

\end{thebibliography}

\appendix
\section{Proofs}\label{sec:proofs}

\paragraph{Proposition~\ref{prop:aff}.}
 $$\forall i ~\forall ~e_1,e_2 \in B_i:~~  \frac 12 \cdot \left(\frac{d_{e_1}}{d_{e_1
          e_2}}\right)^\alpha < \hat a_{e_1}(e_2)  <2 \cdot \left(\frac{d_{e_1}}{d_{e_1
          e_2}}\right)^\alpha \:,$$
    $$\forall ~e_1,e_2 \in \LL:~~ \hat a_{e_1}(e_2) = \left(\frac{d_{e_2}}{d_{e_1
          e_2}}\right)^\alpha \text{ in the uniform power model.}$$

\begin{proof}
    Recall that $\hat a_{e'}(e) \eqdf \frac{S_{e'e}}{S_e}$, $S_e\eqdf P_e/d_e^\alpha$, and
    $S_{e'e}=P_{e'}/{d_{e'e}^{\alpha}}$.
    Note that every two links $e_1,e_2 \in B_i$, satisfy that $S_{e_1} / S_{e_2} \in (1/2,2)$.
    Hence,
    \begin{eqnarray*}
        \hat a_{e_1}(e_2)  & = & \frac{S_{e_1e_2}}{S_{e_2}} = \frac{S_{e_1e_2}}{S_{e_1}} \cdot
                                    \frac{S_{e_1}}{S_{e_2}}\\
                        & = & \frac {P_{e_1}/d_{e_1e_2}^{\alpha}}{P_{e_1}/d_{e_1}^{\alpha}} \cdot
                            \frac{S_{e_1}}{S_{e_2}} \\
                        & = &\left(\frac {d_{e_1}}{{d_{e_1e_2}}}\right)^{\alpha} \cdot \frac{S_{e_1}}{S_{e_2}}\:,
    \end{eqnarray*}
    as required.

    On the other hand, in the uniform power model assignment, all links transmit with the
    same power, namely $P_e = P_{e'}$ for every two links $e$ and $e'$. Hence,
    \begin{eqnarray*}
        \hat a_{e_1}(e_2)  & = & \frac{S_{e_1e_2}}{S_{e_2}} \\
                        & = & \frac {P_{e_1}/d_{e_1e_2}^{\alpha}}{P_{e_2}/d_{e_2}^{\alpha}} \\
                        & = &\left(\frac {d_{e_2}}{{d_{e_1e_2}}}\right)^{\alpha}\:,
    \end{eqnarray*}
    as required.
\end{proof}

\paragraph{Theorem~\ref{thm:out}}
    \textit{Let $L$ denote an SINR-feasible set of links. If $L \subseteq
  B_i$, then
     \[
\forall e\in B_i:~~~           \sum_{\{e' \in L : d_{e'} \geq d_{e}\}}\bar a _{e}(e') = O(1).
     \]}
\begin{proof}
    Theorem 1 in ~\cite{K10} implies that
    $$\sum_{\{e' \in L : d_{e'} \geq d_{e}\}}\min\left\{1,\left(\frac{d_e}{d_{ee'}}\right)^{\alpha}\right\}
        +\sum_{\{e' \in L : d_{e'} \geq d_{e}\}}
    \min \left\{1,\left ( \frac{d_e}{d_{s_{e'}r_e}}\right)^{\alpha}\right\} = O(1).$$
    It follows that,
    \begin{eqnarray*}
        O(1) & = & \sum_{\{e' \in L : d_{e'} \geq d_{e}\}}\min\left\{1,\left(\frac{d_e}{d_{ee'}}\right)^{\alpha}\right\} \\
        & \geq & \sum_{\{e' \in L : d_{e'} \geq d_{e}\}}\min\left\{1,\frac 12 \cdot{\hat a}_{e}(e')\right\} \\
        & = & \sum_{\{e' \in L : d_{e'} \geq d_{e}\}}\min\left\{1,\frac {1}{2\cdot \gamma_{e'}} \cdot a_{e}(e')\right\} \\
        & \geq & \sum_{\{e' \in L : d_{e'} \geq d_{e}\}}\min\left\{1,\frac {\eps}{2\cdot (1+\eps) \cdot \beta} \cdot a_{e}(e')\right\}\:,
    \end{eqnarray*}
    where the second line follows since $L \subseteq B_i$ and Proposition~\ref{prop:aff}. The third line follows from the definition of $a_e(e')$. The last line follows from Proposition~\ref{prop:gamma}.
    The theorem follows, since $\frac {\eps}{2\cdot (1+\eps) \cdot \beta} = O(1)$ and since $\bar a_{e'}(e)\eqdf \min\{1,a_{e'}(e)\}$.
\end{proof}

\end{document}